\documentclass{article}
\usepackage{amsmath,amsthm,amsfonts,amssymb,color,hyperref,anysize,enumitem,graphicx,epstopdf,algorithm,algorithmic,multirow}
\usepackage[usenames,dvipsnames]{xcolor}
\usepackage[margin=0.6in]{geometry}

\newcommand{\rr}{\mathbb{R}}

\newcommand{\ra}{\rightarrow}

\newcommand{\expect}[1]{\mathbb{E}\left[ #1 \right]}
\newcommand{\expectsub}[2]{\mathbb{E}_{#1}\left[ #2 \right]}

\DeclareMathOperator*{\argmin}{arg\,min}
\DeclareMathOperator*{\argmax}{arg\,max}

\usepackage{framed}

\usepackage{tcolorbox}
\tcbuselibrary{theorems}

\newtcbtheorem[auto counter]{coloredDEF}{Definition}
{colback=blue!5,colframe=blue!40!black,fonttitle=\bfseries,before={\vspace{0.3cm}},after={\vspace{0.3cm}}}{de}

\newtcbtheorem[auto counter]{coloredNOT}{Notation}
{colback=blue!5,colframe=blue!40!black,fonttitle=\bfseries,before={\vspace{0.3cm}},after={\vspace{0.3cm}}}{no}

\newtcbtheorem[auto counter]{assume}{Assumption}
{colback=green!5,colframe=green!60!black!80,fonttitle=\bfseries,before={\vspace{0.3cm}},after={\vspace{0.3cm}}}{assume}

\newtcbtheorem[auto counter]{coloredLEM}{Lemma}
{colback=red!5,colframe=red!60!black,fonttitle=\bfseries,before={\vspace{0.3cm}},after={\vspace{0.3cm}}}{lm}

\newtcbtheorem[auto counter]{coloredTHM}{Theorem}
{colback=purple!5,colframe=purple!50!black,fonttitle=\bfseries,before={\vspace{0.3cm}},after={\vspace{0.3cm}}}{th}

\newtcbtheorem[auto counter]{coloredPROP}{Proposition}
{colback=purple!5,colframe=purple!50!black,fonttitle=\bfseries,before={\vspace{0.3cm}},after={\vspace{0.3cm}}}{pr}

\newtcbtheorem[auto counter]{coloredCOR}{Corollary}
{colback=purple!5,colframe=purple!50!black,fonttitle=\bfseries,before={\vspace{0.3cm}},after={\vspace{0.3cm}}}{co}

\newtheorem{theorem}{Theorem}
\newtheorem{lemma}[theorem]{Lemma}
\newtheorem{corollary}[theorem]{Corollary}

\newtheorem{example}[theorem]{Example}

\newtheorem{prop}[theorem]{Proposition}
\newtheorem{proposition}[theorem]{Proposition}
\newtheorem{definition}[theorem]{Definition}
\newtheorem{shortexercise}[theorem]{Short Exercise}
\newtheorem{exercise}[theorem]{Exercise}

\usepackage{graphicx} %
\usepackage{notations}
\usepackage{natbib}
\usepackage{cleveref}
\usepackage{wrapfig}
\usepackage{tabularx}
\usepackage{authblk}
\usepackage{booktabs}
\title{Entry Barriers in Content Markets}
\author{Haiqing Zhu\footnote{Email addresses: \{haiqing.zhu, yunkuen.cheung, lexing.xie\}@anu.edu.au }}
\author{Lexing Xie}
\author{Yun Kuen Cheung}
\affil{School of Computing, Australian National University}
\date{August 2025}

\begin{document}

\maketitle

\begin{abstract}
The prevalence of low-quality content on online platforms is often attributed to the absence of meaningful entry requirements. This motivates us to investigate %
whether %
implicit or explicit
entry barriers, alongside appropriate reward mechanisms, can enhance content quality. We present the first game-theoretic analysis of two distinct types of entry barriers in online content platforms. The first, a \emph{structural} barrier, emerges from the collective behaviour of incumbent content providers which disadvantages new entrants. We show that both rank-order and proportional-share reward mechanisms induce such a structural barrier at Nash equilibrium.
The second, a \emph{strategic} barrier, involves the platform proactively imposing entry fees to discourage participation from low-quality contributors. We consider a scheme in which the platform redirects some or all of the entry fees into the reward pool. We formally demonstrate that this approach can improve overall content quality.
Our findings establish a theoretical foundation for designing reward mechanisms coupled with entry fees to promote higher-quality content and support healthier online ecosystems.

\end{abstract}
\section{Introduction}

Online content markets, such as those found on platforms like YouTube, Spotify and Medium, are dynamic ecosystems where content is constantly created, consumed, and monetized. A reward mechanism, 
driven by the platform but also influenced by other stakeholders like advertisers and enterprise sponsors, is used to allocate %
revenues among content creators.
A defining characteristic of these markets is their high turnover rate: much of the content (e.g., news, sports highlights) becomes obsolete within a short time frame, leading to rapid cycles of attention and monetized value. The rapid cycles imply that creators are engaged in ongoing competitive games, which has motivated a stream of recent work employing game-theoretic analysis to understand the properties of the outcomes of these games.

A commonly observed outcome on many platforms is the prevalence of low-quality content. This is often attributed to the absence of meaningful entry requirements, which allows participation by contributors engaged in low-effort production. This motivates us to investigate whether implicit or explicit \emph{entry barriers}, alongside appropriate reward mechanisms, can discourage such participation and thereby enhance overall content quality.
Entry barriers have been well-studied in economics and are typically categorized into \emph{structural} and \emph{strategic} barriers. (\cite{west2005barriers}) Structural barriers arise when the game's inherent setup and outcome (e.g., Nash equilibrium) create an environment that deters new entrants. In classical economies, this occurs in Cournot markets when incumbent producers already cover market demand at marginal profitability, and even a significant price reduction would not generate enough additional demand to make entry worthwhile for new producers.

On the other hand, strategic barriers involve a stakeholder, say the platform in a content market or one of the dominant incumbents, \emph{changing the game} by proactively introducing additional incentives %
to make market entry more difficult. One key insight from this work is by strategically imposing entry fees, the platform can deter low-quality contributors; furthermore, by redirecting these fees into the reward pool, the platform has the flexibility to tune the existing reward mechanism in ways that incentivize higher-quality content.%

We present the first game-theoretic analysis of both structural and strategic entry barriers in online content markets, focusing on two widely studied reward mechanisms, the Rank-Order Mechanism (RO) and the Proportional Mechanism (PM). 
In an RO mechanism, the platform (via recommender system) and other stakeholders rank content by quality and allocate rewards based on these rankings. The RO family includes the well-known Top‑$k$ mechanisms.
In contrast, PM rewards creators in proportion to the quality of their content items, a scheme commonly used by online streaming platforms. (\cite{youtubePartnerProgram, TikTokReward}) In \Cref{sec: structural-barrier}, we show that under a fixed and limited total reward, and under mild assumptions on the RO mechanisms, 
the Nash equilibrium strategies adopted by incumbent content creators form a structural barrier that deters new entrants.
For PM, we establish that the marginal cost of producing low-quality content is the decisive factor of whether new market entry is profitable.
In \Cref{sec: strategic-barriers}, we formally define the Entry Fee Reallocation Mechanism (EFRM). We show that two instances of EFRM, which we call Max-Min Reallocation and Max-Max Reallocation, respectively maximize the average quality level and the expected maximum quality level at their corresponding Nash equilibrium outcomes.

The remainder of the paper begins with a discussion of related work, followed by a formal description of our model in \Cref{sec: model}. \Cref{sec: structural-barrier,sec: strategic-barriers} present our analyses of the two types of entry barriers. Due to space constraints, all proofs of formal results are deferred to the appendix.

\negspace
\negspace
\paragraph{Related work} 
On the reward mechanismss of online content platforms, \citet{GH14} studied the equilibrium properties under both RO and PM with unbounded long-term 
rewards, while our work focuses on bounded short-term rewards. \citet{ben2018game} connects reward mechanisms to the design of recommender systems, and proposed a Shapley-value-based mechanism that enables both stability and fairness. %
Following this, game-theoretic analyses were conducted to characterise equilibria under different %
mechanisms, including Top-$1$ (\cite{jagadeesan2023supply}), Top-$k$ (\cite{yao2023howBad}), generalised Shapley-value-based approach (\cite{yao2024rethinking}), and probabilistic approaches (\cite{hron2023Modeling, qian2024digital}). RO (studied in this work) is a strict %
generalization of Top-$1$ and Top-$k$ mechanisms.
The creator competition games induced by PM rewards are %
equivalent to the classical%
~\cite{tullock1967welfare} contest, which was recently analysed by \cite{yao2024unveiling,yao2024human}.
Under an auction setting, \citet{mirrokni2010quasi} studied the quasi-proportional mechanism, which can be viewed as a special case of PM.

In theories of competition in economics, barrier to entry generally means \emph{an impediment that makes it more difficult for a firm to enter a market}. (\cite{mcafee2004whatIsBarrierToEntry,west2005barriers})
While it is sometimes blamed for causing monopoly and oligopoly, in other settings it is found to be effective in protecting consumers and the public.
In the context of content creator competition, \citet{ghosh2011incentivizing} proposed a creator selection and elimination mechanism to elicit both high quality content and high user participation. More recently, \citet{hossain2024dataMarket} demonstrated that charging transaction costs in data markets can lead to stronger welfare guarantees.

Competitions among production firms have been a fundamental topic in economics. The first analytical model can be traced back to the \cite{cournot1838recherches} competition in the 19th century. In classical economies, the main point of query is how the competition induces the production quantities and prices at equilibrium. In contrast, content in online platforms can be replicated indefinitely at essentially no cost and hence zero-priced, and attention becomes the currency of trading.

\section{Model} \label{sec: model}
Consider a content market with a set of  creators
$[n] = \{1,2,\ldots,n\}$. Each creator $i$ produces an item with a  chosen \emph{quality level} $q_i \in [0,1]$, %
and they have a cost function $c_i: [0,1] \ra \rr_{\geq 0}$, which maps the quality level of their item to the cost of production. 
Economic intuitions let us set $c_i(0) = 0$. We further assume that $c_i$ is strictly increasing
and twice differentiable.

Upon receiving the items produced by the creators, the platform (via users) determine the reward for each creator.
A \emph{reward mechanism} is a function $\rAll: [0,1]^n \ra \rr_{\geq 0}^n$,
which maps quality values of all items to rewards for each creator.
Following standard game-theoretic notation, the utility of creator $i$ is the difference between their reward and their cost:
\negspace
\negspace
\[
u_i(q_i, \qAll_{-i}) = r_i(q_i, \qAll_{-i}) - c_i(q_i)
\]
\negspace
\negspace
where $\qAll_{-i} = (q_1, \ldots,q_{i-1}, q_{i+1}, \ldots, q_n)$ denotes the item quality from all other creators.

The normalised quality level $q_i \in [0,1]$ can be interpreted as the probability of user engaging with content $i$ after being exposed. (\cite{salganik2006experimental,zhu2023stability}) Despite being simple, it captures the main intuitions in analyzing personalized quality vectors, namely, the probability of engagement being the generalised inner product between quality and user preference~\cite{yao2023howBad}, and the magnitude of quality vector correspond to creator effort~\cite{jagadeesan2023supply}.

\paragraph{Nash Equilibrium (NE)} 
One can view this content market
as a non-cooperative game among the creators. %
The action space of each creator is the domain of quality level $[0,1]$.
Let $\Delta([0,1])$ be the space of probability distributions over $[0,1]$. A mixed strategy profile $\mathbf{F} = (F_1, \ldots, F_n) \in \Delta([0,1])^\numItems$ denotes a quality distribution for each creator, where each $F_i$ is a CDF over the domain $[0,1]$.

\begin{definition}
    A profile $\mathbf{F}^\star$ is a mixed strategy Nash equilibrium when for every $i\in \mathcal{I}$,
    \negspace
    \[
 F^\star_i  ~\in~ \argmax_{F_i \in \Delta([0,1])} ~ \expectsub{(q_i,\q_{-i})\sim (F_i,\mathbf{F}^\star_{-i})}{u_i(q_i, \q_{-i})}.
    \]
    \negspace
    $\mathbf{F}^\star$ is a pure Nash equilibrium if each $F_i^\star$ is a point mass 
    of probability $1$.
    \label{def: Nash Equilibrium}
\end{definition}

The online content market features rapid cycles of attention. Accordingly, creators typically adopt an agile production pipeline that allows for mobility and flexibility. This justifies the use of mixed NE as a solution concept, where creators can randomize and experiment across multiple attention cycles.
We consider two reward mechanisms throughout this work.

\paragraph{Rank-Order (RO) Mechanism} 
assigns monotone rewards $\alpha_1 \geq \alpha_2 \geq \ldots\geq \alpha_n \geq 0$ so that the $i$-th ranked creator by quality will receive reward $\alpha_i$.
Formally, for the permutation $\sigma: \{1,\ldots,\numItems\} \rightarrow \{1,\ldots,\numItems\}$ satisfying $q_{\sigma(1)} > \ldots > q_{\sigma(\numItems)}$, with ties broken at random, the reward is $r_i^{\text{RO}}(q_i,\q_{-i}) = \alpha_{\sigma^{-1}(i)}$. %
Without loss of generality, we assume bounded total reward $\sum_{i\in \itemSet} \alpha_i \leq 1$. 
RO is a broad family, which include the Top-$k$ mechanism widely used in practice, where the best $k$ creators share all the rewards, and $\alpha_j = 0$ for $k+1\le j\le n$. %

\paragraph{Proportional Mechanism (PM)} assigns the rewards proportional to their quality, i.e.,
$r_i^{\text{PM}}(q_i, \q_{-i}) = \frac{q_i}{\sum_{j\in \itemSet}q_j}$.
We note that content market with PM rewards are mathematically analogous to the famous Tullock contest,
with quality in content market corresponding to quantity of production in Tullock contest.
Proportional allocation traces back to mechanisms for network rate control~\cite{kelly1998rate} and systems scheduling~\cite{stoica1996proportional}, and applied to user-generated content~\cite{ghosh2011incentivizing}. It closely resembles the multinomial logit choice model in economics, and often seen as a natural and fair mechanism.

Most recent work (see appendix) assumes one of the mechanisms is at play
but the different properties in the two mechanisms lead to very different conclusions. 
\cite{GH14} was the first to identify RO and PM leads to different content qualities.
This work is the first that discusses entry barriers of both mechanisms in online platforms.

\section{Structural Barriers} \label{sec: structural-barrier}

In this section, we characterise the NE of both mechanisms and study the structural entry barriers incurred by the inherent equilibrium structures of the mechanisms. Specifically, under RO, given the reward structure unchanged, we show that there is no incentive for a new creator to join the platform. Under PM, the ability of a creator to produce low-quality content is the determining factor of whether the creator can enter the market.

\subsection{Structural Barriers under Rank-Order Mechanism} \label{sec: RO}

It is informative to study the properties of NE under RO before presenting the main result on structural barriers.
Denote the vector of {\it descending} rewards as $\Alpha = (\alpha_1,\alpha_2,\ldots,\alpha_n)$; we assume $\alpha_n = 0$ wlog.
When all creators $i\in [n]$ have the same cost function $c_i = c$, the market have a symmetric mixed strategy NE; we denote the CDF of the quality that each creator chose as $\ROcdf$. That is, all creators have the identical competitive advantage, and who ranks on top to get higher rewards is down to a random draw according to the mixed strategy. We define a helper function $h_n(\Alpha,t) = \sum_{i=0}^{n-1} \alpha_{i+1}\binom{n-1}{i}t^{n-1-i} (1 - t)^{i}$, where quantile value $t\in [0,1]$. One can see that each term in this function corresponds to the probability item at the $t^{\text{th}}$ quantile being ranked at the $(i+1)^{\text{st}}$ position and receive reward $\alpha_{i+1}$.
Another critical quantity is $c(1)$, the cost incurred to create the {\it perfect} content at $q=1$. The magnitude of rewards relative to $c(1)$ drives creators' strategies and utilities. 
\citet{GH14} demonstrated the existence of symmetric mixed strategy NE. \Cref{prop: RO-equilibrim-eq} explicitly computes this NE.

\begin{restatable}{theorem}{ROChar} When all creators have the same cost function $c$, the symmetric mixed NE under RO follows one of the three cases:
    \begin{enumerate}[leftmargin=0.2in]
        \item If $c(1) \ge \alpha_1$, $F$ is determined by the functional equation $h_n(\Alpha,F(q)) = c(q)$,
        where the LHS, $h_n(\Alpha,F(q))$, is the expected reward of each creator.
        The support of $F$ is $[0,q_{\max}]$, where $q_{\max} = c^{-1}(\alpha_1)$.
        The expected profit of each provider is $0$.
        \item If $c(1) \le \frac 1n \sum_{i=1}^n \alpha_i$, $F$ has a point mass at $q=1$ of probability $1$.
        The expected profit of each creator is $\left(\frac 1n \sum_{i=1}^n \alpha_i\right) - c(1)$.
        \item If $\frac 1n \sum_{i=1}^n \alpha_i < c(1) < \alpha_1$, $F$ has a point mass at $q=1$ of probability $y$,
        where $y$ satisfies the equality $h_n(\Beta,1-y) = c(1)$, with $\Beta = (\alpha_1,\frac{\alpha_1+\alpha_2}{2},\frac{\alpha_1+\alpha_2+\alpha_3}{3},\ldots,\frac{\sum_{i=1}^n \alpha_i}{n})$. $F$ also has support on $[0,\hat{q}]$, where $\hat{q} = c^{-1}(h_n(\Alpha,1-y))$. On this interval, $F$ is determined by the functional equation $h_n(\Alpha,F(q)) = c(q)$.
        The expected profit of each provider is $0$.
    \end{enumerate}
\label{prop: RO-equilibrim-eq}\negspace
\end{restatable}

\paragraph{Example: NE of a 3-creator market with (different) linear costs.}
Top-2 creators are rewarded with $\alpha_1 = \alpha_2 = \frac 12$ and $\alpha_3 = 0$.
\begin{enumerate}[leftmargin=0.2in]
\item If $c(q) = q$, by \Cref{prop: RO-equilibrim-eq} (1), the NE is obtained by solving $h_3(\Alpha,F(q)) = q$, giving $F(q) = 1 - \sqrt{1-2q}$ for $0\le q \le \frac 12$.
\item If $c(q) = \frac{q}{4}$, by \Cref{prop: RO-equilibrim-eq} (2), the NE is pure, all created content are at $q=1$.
By breaking ties randomly, the expected reward of each provider is $\frac 13 (\alpha_1 + \alpha_2 + \alpha_3) = \frac 13$.
\item If $c(q) = \frac{2q}{5}$, by \Cref{prop: RO-equilibrim-eq} (3), the probability $y$ that a creator produces a perfect content satisfies the equation $\frac {(1-y)^2}{2} + y(1-y) + \frac{y^2}{3} = \frac 25$, giving $y = \sqrt{0.6}\approx 0.7746$. Then $\hat{q} = \frac 54 (1-\sqrt{0.6})^2 + \frac 52 \sqrt{0.6}(1-\sqrt{0.6}) = 0.5$. When $0\le q \le 0.5$, by solving $h_3(\Alpha,F(q)) = \frac{2q}{5}$, we have $F(q) = 1 - \sqrt{1-\frac{4q}{5}}$.
\end{enumerate}

In the third case, the content pool splits into two distinct groups: one consisting of perfect items that target high rewards, and another consisting of lower-quality items (with quality at most $0.5$) that earn less rewards but incur lower costs.
This gap emerges due to the relatively low cost of producing perfect content, which incentivizes creators of near-perfect items to make slight improvements in quality to gain
significantly  higher rewards. In contrast, producing content within the intermediate quality range is sub-optimal, since such items cannot compete effectively with perfect content, yet their production costs are unnecessarily high to outperform low-quality, low-cost alternatives.

In the rest of this paper, we focus on the scenarios captured by the first case of \Cref{prop: RO-equilibrim-eq}, where rewards are limited relative to the cost, and creators are not incentivized to produce (almost) perfect contents, which are commonly observed in hectic content cycles of many platforms. %
Each different $\Alpha$ correspond to a different RO mechanism and thus different NE outcome. We are interested in how the platform might choose
$\Alpha$ that leads to the \emph{optimal} %
quality of contents. A natural benchmark is the expected quality of each content being created, i.e., $\expectsub{q\sim F}{q}$. However, from the social/user perspective, it is probably better to have a small number of stand-out high-quality contents, rather than a lot of mediocre contents in the platform. This suggests another natural benchmark, namely the expected maximum quality among all items created, i.e., $\expectsub{(q_1,q_2,\ldots,q_n)\sim (F,F,\ldots,F)}{\max_i q_i}$.
Mathematically, these two benchmarks correspond to $L^1$-norm and $L^\infty$-norm of $\mathbf{q} = (q_1,q_2,\ldots,q_n)$, so we also consider a smooth interpolation between them, namely $L^p$-norm for $1\le p\le \infty$.

It turns out that when the cost function is linear,
the $L^1$-optimal mechanism corresponds to ``the final elimination principle'', %
where the Top-$(n-1)$ creators share all the rewards evenly, but the last provider gets no reward. %
In stark contrast, $L^\infty$-optimal mechanism correspond to Top-$1$, or ``winner takes all''.
The next two propositions depict optimal mechanisms with more general cost functions.

\begin{restatable}{proposition}{ROOptimalMech}
Under the bounded reward $\sum_{i=1}^n \alpha_n \leq 1$,
the $L^1$-optimal mechanism %
is given by $\alpha_1 = \alpha_2 = \ldots = \alpha_{n-1} = \frac{1}{n-1}$ and $\alpha_n = 0$,
as long as the cost function $c$ is convex.  \label{prop: RO-optimal-mech-ori}
\end{restatable}

\negspace
\begin{restatable}{proposition}{ROMaxOptSol}
Under bounded total reward $\sum_{i=1}^n \alpha_n \leq 1$, %
the $L^\infty$-optimal mechanism 
 is given by $\alpha_1 = 1$ and $\alpha_2 = \alpha_3 = \ldots = \alpha_{n} = 0$, so long as the cost function satisfies $c''(q) \leq c'(q)^2$.
 \label{prop: RO-optimal-max} 
\end{restatable}
We note that the assumption $c''(q) \leq c'(q)^2$ is satisfied for any concave functions, linear functions, and strictly convex functions with $(-1/c')' < 1$ such as $c(q) = (q+1)^2 - 1$.

We also show that when $p$ is sufficiently small, the $L^p$-optimal mechanism remains Top-$(n-1)$, while for all sufficiently large but finite $p$, the $L^p$-optimal mechanism is Top-$1$.
The optimal mechanism between these will further depend on the specific value of $p$ and the cost function.
\begin{restatable}{proposition}{ROOptimalMechLp}
    Suppose the total resource $\sum_{i=1}^n \alpha_n \leq 1$, $(-1/c')' < 1$ and $c(1) < \infty$. There exists $p_1,p_2$ which satisfy $1 \leq p_1 < p_2 < \infty$, such that (i) if $1 \leq  p \leq p_1$, the $L^1$-optimal mechanism %
    is the final elimination mechanism, and
    (ii) if $p \geq p_2$, the $L^p$-optimal mechanism is the Top-1 mechanism.
    \negspace \label{prop: RO-opt-mech-Lp}
\end{restatable}

\Cref{prop: RO-opt-mech-Lp} validates the power of general Top-$k$ mechanisms by establishing that both Top-1 and Top-$(n-1)$ mechanisms achieve $L^p$-optimality for different ranges of $p$.

To identify the entry barriers under RO, we consider the scenario when a new creator with the same cost function $c$ tries to enter the market.

\begin{theorem}[Structural Barrier of RO Mechanism]\label{thm: structural-barrier-RO}
Suppose the following holds:
\begin{itemize}[leftmargin=0.2in]
\item All creators have the same cost function $c$, with $c(1) > 1$.
\item Under a RO mechanism for $n$ creators with rewards $(\alpha^n_1,\alpha^n_2,\ldots,\alpha^n_n)$, let $F_n$ be the CDF of the symmetric NE.
\item If a new creator joins the market, a new RO mechanism for $(n+1)$ creators with rewards $(\alpha^{n+1}_1,\alpha^{n+1}_2,\ldots,\alpha^{n+1}_n,\alpha^{n+1}_{n+1})$ is adopted, and it satisfies $\alpha^n_i\ge \alpha^{n+1}_i$ for $1\le i \le n$, and $\alpha^{n+1}_{n+1} =0$.
\end{itemize}
     Let $g_{n+1}(q)$ be the expected reward of the new creator under the new RO mechanism when she produces a content of quality level $q$, while the other $n$ creators produce contents of quality level according to $F_n$.%
    We have
    \[
    g_{n+1}(q) \leq c(q),~~~~~\forall q\in [0,1],
    \]
    where the inequality is strict when $F_n(q)\in (0,1)$. In other words, there is no way the new creator can enter the market making a positive expected profit upon the established NE of the incumbents.
\label{thm: structural-barrier-RO} 
\end{theorem}

\Cref{thm: structural-barrier-RO} applies if the RO mechanism in use is Top-$k$ whenever the number of creators is at least $k+1$.
However, the theorem does not apply if the RO mechanism in use is Top-$(n-1)$ whenever $n$ creators are in play.

We explain why the structural barrier helps prevent a collapse of overall content quality.
Suppose, to the contrary, that no such barrier exist for any number of incumbent creators.
Then the number of creators, denoted by $n$, can increase indefinitely.
However, since the total reward is bounded by $1$, each creator receives at most $\frac 1n$ reward in expectation.
To avoid loss, their expected cost is also bounded by $\frac 1n$, i.e, $\expectsub{q\sim F}{c(q)}\le \frac 1n$.
If the cost function $c$ is convex, then $c(\expectsub{q\sim F}{q}) \le \expectsub{q\sim F}{c(q)}\le \frac 1n$ by the Jensen's inequality,
which implies $\expectsub{q\sim F}{q} \le c^{-1}(1/n)$. Since $\lim_{n\ra \infty} c^{-1}(1/n) = 0$, the average content quality,
i.e., the $L^1$ metric $\expectsub{q\sim F}{q}$, would vanish in the limit.

\subsection{Proportional Mechanism} \label{sec: PM}
Under the proportional mechanism, we study the case that the creators' cost functions are heterogeneous. We show that, in the equilibrium, entry barrier only depends on a creator's ability to create low quality content, which is governed by each creator’s marginal cost at zero quality, i.e., $c'(0)$. Throughout this section, we assume the cost functions $c_i$ are convex.

\begin{restatable}{theorem}{NoMixEqProp}[\citet{SZIDAROVSZKY1997135}]
    Under the proportional mechanism, there exists a unique pure Nash equilibrium. 
    \label{proposition: prop-no-mix-eq}
\end{restatable}

We would like to classify the participating creators who are really ``contributing'',  producing items with strictly positive quality instead of zero quality items (which is equivalent to non-producing). We formally define this below.
\begin{definition}
    In the Nash equilibrium under proportional mechanism $(q_i^\star)_{i\in \itemSet}$, the creator $i$ is said to be contributing if $q_i^\star > 0$. Moreover, the equilibrium is called a contributing equilibrium if every creator is contributing.
\end{definition}
Given a NE with asymmetric cost functions, there can be some creators who are contributing but the others are not. This corresponds to the fact that the contributing ones are establishing an entry barrier against the others, such that the other creators are not motivated to contribute. To better see this, denote the set of contributing creators as $\mathcal{C}$, we note that, for any $k \in \itemSet \setminus \mathcal{C}$,
\begin{equation}
(\uPO_k)'(0, \q_{-k}^\star) = \frac{1}{\sum_{i\in \mathcal{C}} q_i^\star } - c_k'(0) < 0. 
\label{eq: barrier-to-entry}
\end{equation}
From the above equation, we can see that the non-contributing creators cannot enter the market since $c_k'(0)$ is too large such that it exceeds the ``barrier'' $\frac{1}{\sum_{k\in \mathcal{C}} q_k^\star }$ established by the contributing ones.

\paragraph{Example: A simple market with 3 creators} Consider a market with 3 creators where $c_1(q) = c_2(q) = \frac{1}{2}q^2$ and $c_3(q) = \frac{1}{2}q^2 + 4q$. We first consider the Nash equilibrium when the market only contains creators $1$ and $2$. The equation system to solve the Nash equilibrium $(q_1^\star, q_2^\star)$ is 
\negspace
\begin{align*}
    \frac{q_1^\star}{(q_1^\star + q_2^\star)^2} - q_2^\star = \frac{q_2^\star}{(q_1^\star + q_2^\star)^2} - q_1^\star = 0.
\end{align*}
\negspace
By solving the above, we have $q_1^\star = q_2^\star = \frac{1}{2}$. However, one could also check that, with all 3 creators, the equilibrium is $(\frac{1}{2}, \frac{1}{2}, 0)$, which means that the creator $3$ is not contributing. Consider the derivative of utility for creator $3$ at the equilibrium,
\negspace
\[
 (\uPO)'(q_3, \q_{-3}^\star) = \frac{1}{(q_3 + 1)^2} - c_3'(q_3).
\]
One could see that this is decreasing in $q_3$, and even when $q_3 =0$, the value is negative. This means that the optimal strategy for creator $3$ is $q_3^\star = 0$. Indeed, in this case, creator $1$ and $2$ are creating an entry barrier such that creator $3$ could not enter the market. 

We start with the following characterisation of the contributing equilibrium. 
\begin{restatable}{proposition}{PMExist}
    The Nash equilibrium under the proportional mechanism is a contributing equilibrium if
    \negspace
    \[
     \frac{c_j'(0)}{\sum_{i\in \itemSet} c_i'(0)} < \frac{1}{n-1},~~~\forall j \in \itemSet,
    \]
    \negspace
    or $c_j'(0) = 0$ for all $j\in \itemSet$. Moreover, at the contributing equilibrium, it holds that
    \negspace
    \[
    \frac{c_j'(q_j^\star)}{\sum_{i \in \itemSet} c_i'(q_i^\star)} < \frac{1}{n-1},~~~\forall j \in \itemSet.
    \]
    \negspace
    \label{proposition: NE-exist-PM}
\end{restatable}
\negspace
\negspace
The sufficient condition for contributing equilibrium requires the cost functions are not ``too different'' at the origin. In particular, we note that 2 creators or homogeneous creators of an arbitrary amount can always facilitate a contributing equilibrium. For the necessary condition, when $n$ increases, one can observe that it is also stricter according to the pigeonhole principle, which characterises the hardness for plenty of heterogeneous creators to coexist under such a reward mechanism.

Beyond the equilibria where all the creators are contributing, given an arbitrary NE, we can explicitly point out the contributing creators in the general equilibrium, which is given in the following theorem.
\begin{restatable}{theorem}{ContributingSubset}
    Suppose $c_1'(0) \leq \ldots \leq c_n'(0)$, for a Nash equilibrium $(q_i^\star)_{i\in \itemSet}$, the contributing creators in the equilibrium are $\{1, \ldots, k\}$ where $k$ is the maximal index in the sense that \cref{eq: barrier-to-entry} is satisfied for index $k+1$ onwards.\label{prop: contributing-subset}
\end{restatable}
According to \Cref{prop: contributing-subset}, the participation of the creators depends on the cost of production when the item quality is around 0. Specifically, since $c_i(\epsilon) \approx \epsilon \cdot c_i'(0)$ when $\epsilon$ is small, if the creator is able to ``easily'' produce a low-quality item, then the creator is at least able to enter the market and make a (possibly small) profit. %
Further, for a market facilitated by homogeneous creators, when $n \rightarrow \infty$, the barrier to entry established by those creators will converge to $c_i'(0)$. 
\begin{restatable}{corollary}{BarrierHomo} \label{cor: entry-barrier-homogeneous} For a market facilitated by the set of $n$ homogeneous creators $\itemSet_h(n)$ with cost functions $\bar{c}$. For any creator $k$ with $c_k'(0) > \bar{c}'(0) > 0$, there always exists sufficiently large $n$ such that the contributing creators are exactly $\itemSet_h(n)$ in the equilibrium facilitated by $\itemSet_h(n) \cup \{k\}$.   
\end{restatable}

\section{Strategic Barriers}\label{sec: strategic-barriers}

In this section, we examine how the platform can impose entry fees and use them to modify the reward mechanism, so as to improve content quality at NE. The entry fee serves as a strategic entry barrier, deterring players who aim to create low-quality contents to secure minimal profits from entering the competition. Then the platform reallocates the entry fees to increase the rewards.

\subsection{Strategic Barriers under Rank-Order Mechanism}

We will define Entry Fee Reallocation Mechanism (EFRM) formally below, but before doing so, we discuss some natural constraints on how we can modify the reward mechanism.
Suppose the original RO mechanism is $\Alpha = (\alpha_1,\alpha_2,\ldots,\alpha_n)$.
Recall that the rewards can be driven by the platform as well as advertisers and sponsors, so the platform might not enjoy full flexibility to adjust the rewards.
To afford adjustment to the rewards, we propose the platform to charge an entry fee $\xi$ from each creator.
The entry fees collected will be used to modify the RO mechanism to $\overline{\Alpha} = (\overline{\alpha}_1,\overline{\alpha}_2,\ldots,\overline{\alpha}_n)$.
To ensure that no creator incurs a loss, the minimal reward in the modified mechanism should be at least the entry fee charged, i.e., $\overline{\alpha}_n \ge \xi$.
Also, the reward at each rank should not drop, i.e., $\overline{\alpha}_i \ge \alpha_i$ for $1\le i\le n$.
Finally, from the platform's perspective, it is preferable that the increase in total reward not exceeding the total entry fee received, i.e., $\sum_{i=1}^n (\overline{\alpha}_i - \alpha_i) \le n\xi$.

\begin{definition}[Entry Fee Reallocation Mechanism (EFRM)]
Under a ``default'' RO mechanism with rewards $\Alpha$, the platform charges each creator an entry fee of $\xi$.
After collecting the total entry fee of $n\xi$, the platform reallocates it back to the reward pool, forming a new RO mechanism with rewards $\{\overline{\alpha_i}\}_{1\leq i\leq n}$, which satisfy the constraints $\overline{\alpha}_i \geq \max \{\xi,\alpha_i\}$ and $\sum_{i=1}^{n} (\overline{\alpha}_i - \alpha_i) \le n\xi$.
\label{def: EFRM}
\end{definition}

The example below demonstrates how an EFRM can improve the $L^\infty$ metric.

\paragraph{Example: EFRM with 3 creators} Suppose $n=3$, and the original mechanism is $(\alpha_1,\alpha_2,\alpha_3) = (\frac 12,\frac 12,0)$. After an entry fee $\xi$ is charged from each creator,
we adopt EFRM $(\overline{\alpha}_1,\overline{\alpha}_2,\overline{\alpha}_3) = (\frac 12 + \xi + t, \frac 12 + \xi - t, \xi)$, where $0\le t\le \min \{ \frac 12, \xi\}$. The constraint on $t$ is to ensure that all constraints of EFRM are satisfied. After cancelling the effect of entry fees, the EFRM is equivalent to $(\frac 12 + t,\frac 12 - t, 0)$ with no entry fees involved.

Suppose the 3 creators share the same cost function $c(q) = q$. By \Cref{prop: RO-optimal-max}, we can set $\xi = t = \frac 12$ to achieve the $L^\infty$-optimal mechanism. Next, we calculate the improvement at the NE characterised at \Cref{prop: RO-equilibrim-eq}.
For the original mechanism, the NE is attained with $F(q) = 1 - \sqrt{1-2q}$ for $0\le q\le \frac 12$.
For the new mechanism, the NE is attained with $\tilde{F}(q) = \sqrt{q}$ for $0\le q \le \frac 12$.
The $L^\infty$ metrics are respectively $\int_0^{1/2} 1-F(q)^3\,\mathsf{d}q = 0.45$ and $\int_0^1 1-\tilde{F}(q)^3\,\mathsf{d}q = 0.6$.

The following lemma confirm that the EFRM is always beneficial for both benchmarks.
\begin{lemma}
  There always exists an EFRM with $\xi>0$ such that the resulting $L^1$ and $L^\infty$ metrics are not decreased compared to no EFRM. \label{lem: EFRM-no-harm}
\end{lemma}
The key component of EFRM is how to reallocate the entry fees to maximise the $L^1$ and $L^\infty$ metrics. We introduce two different reallocation schemes, Max-Min (\Cref{alg: Max-min reallocation}) and Max-Max (\Cref{alg: Max-max reallocation}). Both reallocation schemes start with remedying the costs of the entry fees. Then, the Max-Min reallocation scheme is to subsidise from the reward of the bottom rank $\alpha_n$ to the top rank $\alpha_1$. Given the constraints stated in \Cref{def: EFRM}, Max-Min is the scheme that maximises the minimum reward $\alpha_n$ given the total charged entry fees. The Max-Max scheme is to allocate all the remaining entry fees to the reward of the top-1 creator $\alpha_1$. 
\begin{algorithm}[H]
    \caption{Max-Min Reallocation} \label{alg: Max-min reallocation}
    \begin{algorithmic}[1]
      \REQUIRE Available entry fees $R$, the initial rewards $\alpha_1, \cdots, \alpha_n$.
      \STATE Use $R$ to remedy all the rewards to $\alpha_i \leftarrow \max\{\xi, \alpha_i\}$. Update $R$.
      \FOR{$i = 1,\cdots, n-2 $}
        \IF{$R > i\cdot(\alpha_{n-i-1}- \alpha_{n-i})$}
          \STATE Allocate the entry fees such that all the rewards below the place $n-i$ matches $\alpha_{n-i-1}$. $\alpha_j \leftarrow \alpha_{n-i-1}$ for all $j \geq n-i$. 
          \STATE Update the remaining entry fees. $R \leftarrow R - i\cdot(\alpha_{n-i-1} - \alpha_{n-i})$.
        \ELSE
          \STATE Allocate all the remaining entry fees in this range: $\alpha_j \leftarrow \alpha_j + R/i$ for all $j \geq n-i$.
        \ENDIF
      \ENDFOR
    \end{algorithmic}
\end{algorithm}

\begin{algorithm}[H]
    \caption{Max-Max Reallocation} \label{alg: Max-max reallocation}
    \begin{algorithmic}[1]
      \REQUIRE Available entry fees $n\xi$, the initial rewards $\alpha_1, \cdots, \alpha_n$.
      \STATE Use $R$ to remedy all the rewards to $\alpha_i \leftarrow \max\{\xi, \alpha_i\}$. Update $R$.
      \STATE Reallocate all the rewards $\alpha_1 \leftarrow \alpha_1 + R$.
    \end{algorithmic}
\end{algorithm}

We show that, under mild conditions, the Max-Min and Max-Max schemes are the optimal reallocation strategies to maximise $L^1$ and $L^\infty$ metrics, respectively.
\begin{theorem} \label{thm: optimal-subsidy-scheme}
Under RO, to implement an EFRM with entry fee $\xi$,
\begin{itemize}[leftmargin=0.2in]
    \item It is optimal to use the Max-Min scheme to reallocate the entry fee to maximise the $L^1$ metric.
    \item  Assume the cost function satisfies the fact that $c''(q) \leq c'(q)^2$, it is optimal to use the Max-Max scheme to reallocate the entry fee to maximise the $L^\infty$ metric.
\end{itemize}
\end{theorem}
The results given in \Cref{thm: optimal-subsidy-scheme} also matches with observations in \Cref{prop: RO-optimal-mech-ori,prop: RO-optimal-max}, where the Max-Min scheme is trying to fix the rewards towards the final elimination mechanism and the Max-Max scheme matches the Top-1 mechanism.

\subsection{Strategic Barriers under Proportional Mechanism} \label{sec: PMC}

For the PM, we also study the case when the platform is charging an entry fee $\xi$ to all creators. Similar to the above, this corresponds to the scenario that $c_i(q) \rightarrow c_i(q) + \xi$.

\begin{definition}[Proportional Mechanism with Entry Fee]
    With entry fee $\xi > 0$, the utilities of the creator $i$ is \[
    u_i(q_i) = \frac{q_i}{\sum_{j\in \mathcal{C}}q_j}- c_i(q_i) - \xi,
    \]
    where the set $\mathcal{C} \subset \mathcal{I}$ is the maximal set of creators with $u_i(q_i^\star) \geq 0$ for all $i \in \mathcal{C}$ given the NE $(q_i^\star)_{i\in \mathcal{I}}$ of PM facilitated by creators in $\mathcal{I}$.  
\end{definition}
Indeed, by charging an entry fee, without considering the positivity of the utilities, the equilibrium strategies remain the same as those under PM by \Cref{eq: barrier-to-entry}. However, the utilities of the creators might be negative. Hence, we only consider the equilibria where all the creators have positive utilities. Also, one could see that the average quality will not decrease even if the entry fee is not redistributed because of the unchanged equilibrium structure. It remains to investigate the set of creators $\mathcal{C}$ who are still contributing in the equilibrium. We start from the following proposition.
\begin{restatable}{proposition}{GoodnessOfCommission}
In the contributing equilibrium under PM, suppose that, for some $i,j \in \itemSet$, $c_i(z) > c_j(z)$ and $c_i'(z) > c_j'(z)$ for any $z \in (0,1)$. Then, $\uPO_i(q_i^\star, \q_{-i}^\star) < \uPO_j(q_j^\star,\q_{-j}^\star)$ and $q_i^\star < q_j^\star$. \label{prop: goodness-of-commision}
\end{restatable}

The above result indicates that charging the entry fees will only expel the creators with the lowest quality items. 
To see this, from \Cref{prop: goodness-of-commision}, once the cost functions are in strict ordering in both zeroth and first-order values, then the profits gained by the creators are also in the same ordering as their cost functions, where lower costs correspond to higher profits. Then, consider the NE under PM, $(q_1^\star,\ldots, q_n^\star)$, with $i,j\in \mathcal{I}$ satisfying the condition indicated in \Cref{prop: goodness-of-commision}, producer $i$ will first be eliminated from the market as $\xi$ increases. 

\begin{wrapfigure}{l}{0.5\textwidth}
  \centering
  \includegraphics[width=0.5\textwidth]{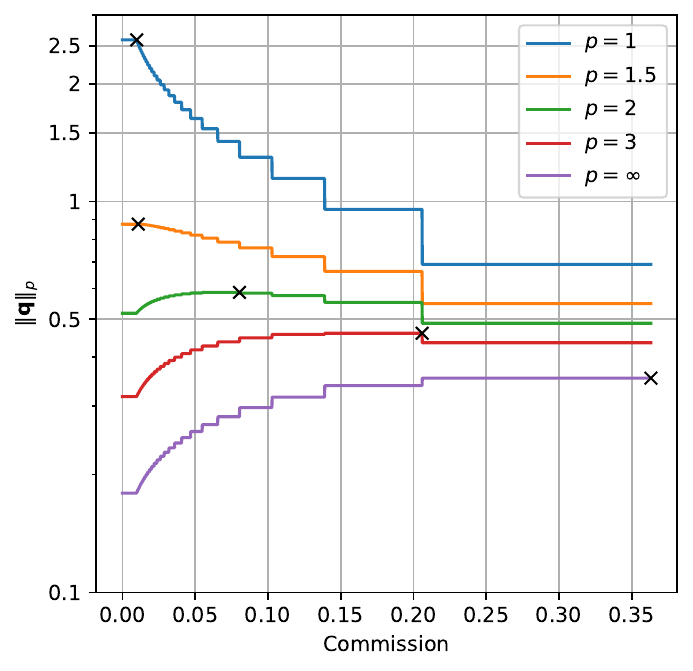}
\negspace
    \caption{Simulation of the proportional mechanism with entry fee of 30 creators. Each curve indicates the $L^p$ norm of the equilibrium qualities $\Vert \q \Vert_p$. The y-axis is in the log scale. The maximum of each curve are marked by ``$\times$''.}\label{fig: commission}
    \bignegspace
\end{wrapfigure}

To further validate the advantage of charging entry fees, we empirically investigate a market instance where the costs $c_i(q) = i\cdot q^2$ and $i\in \{1,\ldots,30\}$. We measure the $p$-norms of the overall qualities at the NE: $\Vert \q^\star \Vert_p$. The result is presented in \cref{fig: commission}. Different values of $p$ indicate different aspects of the market. When $p=1$, the $p$-norm measures the sum of qualities, and when $p = \infty$, the $p$-norm indicates the maximum quality. Different values of $p$ highlight various market aspects. For instance, when $p=1$ the $p$-norm calculates the total sum of qualities, whereas $p=\infty$ reflects the maximum quality. Some market moderators prioritize the aggregate market quality, even if individual item qualities are low. Others might value the highest quality item, despite less diversity. Using the $p$-norm allows us to explore a balance between these priorities and examine how entry fees influence these metrics.

From \Cref{fig: commission}, when $p=1$, a higher entry fee consistently results in a lower total quality. However, for all other values of $p$, as the entry fee increases, the $p$-norm metrics initially increase when the entry fee is relatively low. These metrics peak and then begin to decline. As such, we can conclude that charging a certain nonzero entry fee is generally beneficial, provided the platform is not solely focused on total quality. Nonetheless, excessively high entry fees can be detrimental as they may drive away too many creators from the market. Such phenomenon pinpoints the trade-off between market efficiency and diversity.
Such observations contradict the case in RO, where even subsidising the creators could lead to a decrease in the overall production qualities.

\negspace
\section{Conclusion}

We launch the first game-theoretic analysis of structural and strategic entry barriers in online content markets. Our theoretical findings offer two key insights that may inform the design of practical recommnder systems and reward mechanisms: (1) a carefully designed reward mechanism can induce a structural barrier that prevents collapse of content quality (\Cref{thm: structural-barrier-RO}); and (2) by charging entry fees and redirecting them into the reward pool appropriately, the platform can incentivize the production of higher-quality content (\Cref{thm: optimal-subsidy-scheme}).

Our analysis opens several avenues for further investigation.
First, our theoretical results for rank-order mechanisms focus on the symmetric case where all creators share the same cost function.
In practice, however, creators differ significantly in ability, cost, access to technology, fame and audience reach.
This raises the question: how should platforms design recommender systems and reward mechanisms that account for such heterogeneity while supporting both quality improvement and fairness objectives? 
Second, we have assumed that content quality is directly observable and measurable, whereas in real-world settings, quality assessments are noisy and heavily shaped by algorithmic recommendations. Understanding how such uncertainty interacts with entry barriers is an open problem.
Third, our use of Nash equilibrium  as a solution concept presumes perfect information and static behavior, yet creators often operate with limited feedback and adapt their strategies over time, possibly through learning algorithms. Modeling and analyzing the dynamic interactions between adaptive creators and evolving recommender systems presents a challenging but important direction for future work.
Finally, empirical validation of the relationship between content quality, production costs, and market outcomes using real platform data could provide valuable insights and inform more effective policy and mechanism design.

\bibliographystyle{plainnat}
\bibliography{refs}
\newpage
\appendix

\section{Summary of Related Work} 
\label{app:relatedwork}
In the following table, we summarise the basic setup and their contribution (informally) in equilibrium characterisation of recent literature in modelling the competition in content markets.

\begin{table}[H]
\caption{Comparison of recent work on content markets}\label{tab:relatedwork}
\begin{tabularx}{0.96\textwidth}{ >{\hsize=.2\hsize}X >{\hsize=.28\hsize}X >{\hsize=.22\hsize}X >{\hsize=.32\hsize}X}
\toprule
\textbf{Paper} & \textbf{Reward Mechanism}                & \textbf{Cost Function} & \textbf{Equilibrium Characterisation}    \\ \hline
\cite{GH14}              & Rank-Order + Proportional                                    &Convex, Heterogenoues                  & Existence, asymptotical characterisation                               \\ \hline
\cite{ben2018game}             & Designed monotone reward                                        & Any\footnotemark[1]                       & Potential Game                                       \\ \hline

\cite{jagadeesan2023supply} & Top-1\footnotemark[2] & homogenous - $p$-norm & Existence of Symmetric MNE, Characterisation on genres. \\ \hline
\cite{yao2023howBad}              & Top-K                                    & NA                      & Coarse correlated equilibrium exists, bounds of PoA                                        \\ \hline
\cite{hron2023Modeling}              & Proportional                                   & NA                & Existence of MNE, PNE exists in some cases.                           \\ \hline
\cite{yao2024human}          & Proportional & $c_i x_i^\rho$ for main results   & bounding the micro/macro level costs of productions \\ \hline
\cite{yao2024rethinking}              & Designed monotone reward                                        & Any                      & Potential Game                                       \\ \hline
\cite{yao2024unveiling}          & Proportional                 & Increasing and convex $c(x_i)$, heterogeneous                      & Existence and Uniqueness                               \\ \hline
\cite{yao2024user}              & Proportional + Top-K\footnotemark[3]                                    &NA                  & Empirically - local NE                                   \\ \hline
\textbf{Our work}   & \textbf{Rank-Order + Proportional} & \textbf{Convex, differentiable} & \textbf{Existance and characterisation of structural and strategic entry barriers.} \\ \hline
\end{tabularx}
\end{table}
\footnotetext[1]{The cost function here is not very relevant since the main theme is to deduce a potential game structure}
\footnotetext[2]{ Rank-Order mechanism with $\alpha_1 =1$ and $\alpha_k = 0$ for any $k> 1$.}
\footnotetext[3]{Top-K refers to the reward mechanisms that only give rewards to the Top-K quality items. }

Additionally, we give the following remarks on the differences between our works and existing literature: 
\begin{itemize}
    \item \cite{yao2024human,yao2024unveiling} also adopted the Tullock contest model (equivalent to our proportional mechanism). Our focus in this paper is the characterisation of the structural and strategic entry barriers among the PM mechanism. To the best of our knowledge, it is not discussed in previous works. 
    \item Our RO mechanism is a strong generalisation of the Top-1 mechanism utilised in \cite{jagadeesan2023supply}. Also, our focus is to characterise the possibly optimal mechanism and identify the entry barriers, which is not discussed in the previous works.

    \item Our model differentiates and extends the setup of \citet{GH14} for several perspectives. First, for the cost functions, we do not impose any global assumptions other than differentiability and $c_i(1) > 1$ ( the total resource in the market (which is normalised to $1$ without loss of generality) is not enough to produce a ``perfect'' item). Second, our model mainly considers the scenarios where the total number of producers and total resources are constrained, which better matches the real-world scenarios.

    \item Fundamentally, the game under PM is one instance of concave games \cite{rosen1965exist}, which is studied extensively over the past few decades, including \cite{evendar2009stoc,cai2022finite} on the no-regret dynamics to find NE, \cite{bravo2018bandit} on solving NE under bandit feedback setting.
\end{itemize}

\section{Proof of \Cref{prop: RO-equilibrim-eq}}
\label{app-sec: Proof-of-RO-equilibrim-eq}
We start with the proving there is no PNE exists for RO. Then, we finalise the characterisation of the mixed strategy NE.
\begin{restatable}{lemma}{RONoPNE}
    If $\alpha_1 \geq \alpha_2 \geq \ldots \geq \alpha_n = 0$ and $\alpha_1 > 0$, there does not exist pure Nash equilibria under RO. There exists a symmetric mixed strategy Nash equilibrium without any point mass. \label{lem: RO-No-PNE}
\end{restatable}
\begin{proof}
    Suppose there exists PNE $(q_1^*,\ldots,q_n^*)$. Say $\alpha_{i} > \alpha_{i+1}$. WLOG, if $q_1^* > q_2^* >\ldots > q_n^*$, then we could let $q_{i}' = q_{i}^* - \epsilon$ such that $q_{i}' > q_{i+1}^*$. In this case, the cost for the producer $1$ is lower. In the case that there exists $q_i^* = q_{i+1}^*$ for some $i$, one can set $q_i' = q_i^* + \epsilon$, which also leads to contradiction. The existence of mixed strategy equilibrium is proven in  \cite{GH14}.
\end{proof}

\ROChar*
\begin{proof}
    \textbf{Case 1:} We first note that the best response condition guarantees that, at every point $q$, the expected payoff from the market is the same. 
    
    We first show that the cdfs are atomless. Suppose there is any point mass in $\ROcdf$ at point $p_0 \in [0,1]$. Then, we can select one producer $i$. We change his strategy by moving the point mass at $p_0$ to some point $p_0 + \epsilon$ with $\epsilon$ small enough. Then, consider the event $E = \{\q_{-i} = p_0, q_i = p_0 + \epsilon\}$, it has non-zero positive probability. Hence, it gives a nonzero increase to the expected profit of producer $i$, by choosing $\epsilon$ small enough. Hence, there's no point mass in the symmetric mixed NE.
    
    Next, we show that the point $0$ is included in the NE. Suppose it is not included. Then, one could choose some point in $(0, r_{\min} )$ where $r_{\min}$ is the minimum of the support of the distribution. One could obtain better payoff by moving the mass to that point since it has lower cost. Similar arguments could also be applied to show that there's no ``interval holes'' in the support of the mixed NE. This also shows that the expected profit for every producer is $0$, since the distribution is atomless and includes $0$.

    \textbf{Case 2:} It is immediate to verify this is indeed an NE since any strategy deviation will lead to lower reward $\alpha_n = 0$ which is less than $\frac{1}{n}\sum_{i} \alpha_i$. The profit of this case is then clear.

    \textbf{Case 3:} Given the equilibrium stated above, for any potential strategies that are not in the support. The payoff is $h(\Alpha, 1-y)$ which is less than $h(\Beta, 1-y)$ for any $y$ since $\Alpha \leq \Beta$. Then, using the same argument as the proof of 1, one can again justify that $q=0$ must be in the support of the mixed strategy. Therefore, we have fixed the support of the mixed strategy and the payoff must be $0$, which means that $h(\Alpha, F(q)) = c(q)$.
    
\end{proof}

\section{Proof of \Cref{prop: RO-optimal-mech-ori}}
\label{app-sec: Proof-of-RO-optimal-mech-ori}
In the following proofs, we use the notation $g(q) = h(\Alpha,F(q)) $
The proof relies on the following technical lemma:
\begin{lemma}
 Suppose the total resource is constrained such that $\sum_{i=1}^n \alpha_n \leq 1$. The problem $\max_{\alpha_1, \ldots, \alpha_n} \expectsub{F}{q}$ is equivalent to the following problem
\[
\begin{aligned}
 &\max_{\alpha_1, \ldots, \alpha_n} \int_{0}^{1} c^{-1}(h(\Alpha ,y))\ \mathsf{d}y. \\
 &\text{s.t.}~~ \alpha_1 \geq \ldots \geq  \alpha_n = 0,~~\sum_{i=1}^n \alpha_i =  1.
\end{aligned}
\]
\label{prop: RO-OPT-prob-Ori}
\end{lemma}
\begin{proof}
    First, for any continuous and strictly increasing c.d.f. $y = F(x)$, we note that 
    \[
    \int_{0}^{1} F^{-1}(y)\,\mathsf{d}y =  \int_{0}^{x_{\max}} x\,\mathsf{d}F(x) = \expectsub{F}{x}.
    \]
    Therefore, by \Cref{prop: RO-equilibrim-eq}, we could see that
    \[
    h(\vec{\alpha}, y) = c(F^{-1}(y)) ,
    \]
    where $F$ denotes the c.d.f. of distribution $\rho$ introduced in \Cref{prop: RO-equilibrim-eq}. It implies
    \[
    \expectsub{F}{x} = \int_{0}^{1} F^{-1}(y)\,\mathsf{d}y = \int_{0}^{1} c^{-1}(h(\vec{\alpha} ,y))\ \mathsf{d}y.
    \]
    
\end{proof}

\begin{lemma}[Fast approximation] \label{app-lem: RO-action-improve-mean}
For some $\alpha_i < \alpha_j$, given $\delta$ small enough any action that increases $\alpha_i$ by $\delta$ and decreases $\alpha_j$ by $\delta$ will increase the objective $\expectsub{F}{q}$. 
\end{lemma}

\begin{proof}
  By the Leibniz's rule, we have
  \begin{align*}
  &\frac{\partial }{\partial \alpha_i}\int_{0}^{1} c^{-1}(h(\boldsymbol{\alpha} ,y)) \,\mathsf{d}y - \frac{\partial }{\partial \alpha_j}\int_{0}^{1} c^{-1}(h(\boldsymbol{\alpha} ,y)) \,\mathsf{d}y \\
  &= \int_0^1 \frac{\partial c^{-1}}{\partial h(\boldsymbol{\alpha} ,y)}\cdot \left[\frac{\partial }{\partial \alpha_i}h(\boldsymbol{\alpha} ,y) -\frac{\partial }{\partial \alpha_j} h(\boldsymbol{\alpha} ,y)\right]\,\mathsf{d}y \\ 
   &= \int_0^1 \frac{\partial c^{-1}}{\partial h(\boldsymbol{\alpha} ,y)}\cdot \left[\binom{n-1}{i}y^{n-1-i}(1-y)^{i} - \binom{n-1}{j}y^{n-1-j}(1-y)^{j}\right]\,\mathsf{d}y.
\end{align*}
 Since $i > j$, we know that there exists $y_0$ such that $\binom{n-1}{i}y^{n-1-i}(1-y)^{i} > \binom{n-1}{j}y^{n-1-j}(1-y)^{j}$ for $y < y_0$ and $\binom{n-1}{i}y^{n-1-i}(1-y)^{i} < \binom{n-1}{j}y^{n-1-j}(1-y)^{j}$ for $y > y_0$. Hence, we have
\begin{align*}
&\int_0^1 \frac{\partial c^{-1}}{\partial h(\boldsymbol{\alpha} ,y)}\cdot \left[\binom{n-1}{i}y^{n-1-i}(1-y)^{i} - \binom{n-1}{j}y^{n-1-j}(1-y)^{j}\right]\,\mathsf{d}y \\
&= \int_0^{y_0} \frac{\partial c^{-1}}{\partial h(\boldsymbol{\alpha} ,y)}\cdot \left[\binom{n-1}{i}y^{n-1-i}(1-y)^{i} - \binom{n-1}{j}y^{n-1-j}(1-y)^{j}\right]\,\mathsf{d}y \\ 
&~~~~~~+\int_{y_0}^{1} \frac{\partial c^{-1}}{\partial h(\boldsymbol{\alpha} ,y)}\cdot \left[\binom{n-1}{i}y^{n-1-i}(1-y)^{i} - \binom{n-1}{j}y^{n-1-j}(1-y)^{j}\right]\,\mathsf{d}y \\
&\geq \int_0^{y_0} \frac{\partial c^{-1}}{\partial h(\boldsymbol{\alpha} ,y_0)}\cdot \left[\binom{n-1}{i}y^{n-1-i}(1-y)^{i} - \binom{n-1}{j}y^{n-1-j}(1-y)^{j}\right]\,\mathsf{d}y \\ 
&~~~~~~+\int_{y_0}^{1} \frac{\partial c^{-1}}{\partial h(\boldsymbol{\alpha} ,y_0)}\cdot \left[\binom{n-1}{i}y^{n-1-i}(1-y)^{i} - \binom{n-1}{j}y^{n-1-j}(1-y)^{j}\right]\,\mathsf{d}y \\
&= 0.
\end{align*}
where the inequality is due to fact that $\frac{\partial c^{-1}}{\partial h(\boldsymbol{\alpha},y)}$ is positive and decreasing in $y$. This proves the claim.
\end{proof}
\begin{proof}[Proof of \Cref{prop: RO-optimal-mech-ori}]
    By \Cref{app-lem: RO-action-improve-mean}, $\alpha_1 = \cdots = \alpha_{n-1} = \frac{1}{n-1}$ is the only set of rewards that admit no such manipulation to further increase the benchmark. Hence, it is optimal.
\end{proof}

\section{Proof of \Cref{prop: RO-optimal-max}}
\label{app-sec: proof-RO-optimal-max}
\begin{lemma} \label{app-lem: RO-OPT-max}
Suppose the total resource is constrained such that $\sum_{i=1}^n \alpha_n \leq 1$. The problem $\max_{\alpha_1, \ldots, \alpha_n} \expectsub{F}{\max\{q_1,\cdots, q_n\}}$ is equivalent to the following problem
\[
\begin{aligned}
 &\max_{\alpha_1, \ldots, \alpha_n} \int_{0}^{1} t^{n-1}c^{-1}(h(\Alpha ,t))\ \mathsf{d}t. \\
 &\text{s.t.}~~ \alpha_1 \geq \ldots \geq  \alpha_n \geq 0, ~~\sum_{i=1}^n \alpha_i =  1.
\end{aligned}
\]
\end{lemma}
\begin{proof}
    Note that the cdf of $\max\{q_1,\cdots, q_n\}$ is $F^n(q)$. With the change of variable $y = F^n(q)$, $t = y^{1/n}$, we have
    \[
     \int_{0}^{q_{\max}} q ~\mathsf{d}F^n(q) = \int_{0}^{1} F^{-1}(y^{\frac{1}{n}})~\mathsf{d}y = \int_{0}^{1} F^{-1}(t)~\mathsf{d}t^n = n \cdot \int_{0}^{1} t^{n-1}F^{-1}(t)~\mathsf{d}t.
    \]
    By \cref{prop: RO-equilibrim-eq}, we know that $h(\vec{\alpha}, t) = c(F^{-1}(t))$. Replacing $F^{-1}(t)$ with $c^{-1}(h(\vec{\alpha}, t))$ will yield the results.  
\end{proof}

\begin{lemma} \label{app-lem: RO-action-improve-max}
For some $\alpha_i < \alpha_j$, given $\delta$ small enough any action that decreases $\alpha_i$ by $\delta$ and increases $\alpha_j$ by $\delta$ will increase the objective $\max_{\alpha_1, \ldots, \alpha_n} \expectsub{F}{\max\{q_1,\cdots, q_n\}}$. 
\end{lemma}
\begin{proof}
    Suppose $\Alpha_2$ is the result of such action on some indices $i<j$ to $\Alpha_1$. Utilising the concavity of $c^{-1}$, with exactly the same arguments with \cref{prop: RO-optimal-mech-ori}, the result will follow. The only difference is the usage of the assumption $(c^{-1})'g(\Alpha, t)\cdot t^{n-1}$ is increasing. The inequality of the equation above is reversed. Hence, result is also reversed.
\end{proof}
\begin{proof}[Proof of \Cref{prop: RO-optimal-max}]
    Again, the set of rewards with $\alpha_1= 1$ is the only set of rewards that admit no manipulations mentioned by \Cref{app-lem: RO-action-improve-max}. Hence, it is optimal.
\end{proof}
\begin{lemma}
    $(c^{-1})'h(\Alpha, t)\cdot t^{n-1}$ is increasing as long as for the maximal gap we have $G = \max_{i} \{\alpha_i - \alpha_{i-1}\}$ stastisfies
    \[
       \left(-\frac{1}{c'}\right)' <  \frac{1}{G}.
    \]
\end{lemma}
\begin{proof}
    Taking derivative of the above, we have
    \begin{align*}
     \frac{\partial (c^{-1})'(h(\Alpha, t))\cdot t^{n-1}}{\partial t} &= (c^{-1})''(h(\Alpha, t))\cdot t^{n-1}\cdot (n-1) \expect{gap} + (n-1) (c^{-1})'(h(\Alpha, t))\cdot t^{n-2} \\
     &= (n-1)t^{n-2}\left( (c^{-1})''(h(\Alpha, t))\cdot t\expect{gap} +  (c^{-1})'(h(\Alpha, t)) \right)\\
     &\geq (n-1)t^{n-2}\left( (c^{-1})''(h(\Alpha, t))\expect{gap} +  (c^{-1})'(h(\Alpha, t)) \right).
    \end{align*}
    Say $z = h(\Alpha, t)$, we have the RHS greater than 0 if
    \[
     (c^{-1})''(z)\expect{gap} +   (c^{-1})'(z) \geq 0.
    \]
    By inverse function theore, we have
    \[
     (c^{-1})'(z) = \frac{1}{c'(c^{-1}(z))}
    \]
    and 
    \[
     (c^{-1})''(z) = \frac{\partial}{\partial z}\frac{1}{c'(c^{-1}(z))} = \frac{- c''(c^{-1}(z))}{[c'(c^{-1}(z))]^2} \cdot \frac{1}{c'(c^{-1}(z))} = \frac{- c''(c^{-1}(z))}{[c'(c^{-1}(z))]^3}.
    \]
    Plugging in, we just want
    \[
     \frac{- c''(c^{-1}(z))}{[c'(c^{-1}(z))]^3}\expect{gap} +   \frac{1}{c'(c^{-1}(z))} \geq 0.
    \]
    which is equivalent to the condition that
    \[
     \frac{c''(w)}{c'(w)^2} \leq  \frac{1}{\expect{gap}} .
    \]
    for any $w \in (0,1)$. Since the RHS always greater than $1$, this concludes the proof.
\end{proof}

\section{Proof of \Cref{prop: RO-opt-mech-Lp}}
\label{app-sec: proof-RO-opt-mech-Lp}
\ROOptimalMechLp*
\begin{proof}
    We still start by figuring out the corresponding well-defined optimisation problem. We use the following change of variable $y = F(q)$. Then, the problem could be reformulated as:
    \[
        \begin{aligned}
         &\max_{\alpha_1, \ldots, \alpha_n} \int_{0}^{1} c^{-1}(h(\Alpha ,y))^p\ \mathsf{d}y. \\
         &\text{s.t.}~~ \alpha_1 \geq \ldots \geq  \alpha_n = 0,~~\sum_{i=1}^n \alpha_i =  1.
        \end{aligned}
    \]
    Consider the derivative:
    \[
     \frac{\partial \mathcal{L}}{\partial \alpha_i} =\underbrace{ \int_{0}^{1} p[c^{-1}(h(\Alpha, y))]^{p-1}\frac{\partial c^{-1}(h(\Alpha, y))}{\partial h(\Alpha, y)}\cdot \binom{n-1}{i}y^{i}(1-y)^{n-1-i} ~\mathsf{d}y }_{>0} 
    \]
    Consider the function $c^{-1}(h(\Alpha, y))^p$. The case $p=1$ corresponds the \cref{prop: RO-optimal-mech-ori}. We only need to show that it is increasing whenever $p$ is large enough. Let $z = h(\Alpha, y)$. Consider its second derivative:
    \[
     \frac{\partial^2 [c^{-1}(z)^p]}{\partial z^2} \propto c^{-1}(z)^{p-2} \cdot \left[(p-1)\cdot \frac{\partial c^{-1}(z)}{\partial z} + c^{-1}(z)\cdot \frac{\partial^2c^{-1}(z)}{\partial z^2}\right].
    \]
    Since $c(1) < \infty$, when $p$ is large enough, the second derivative is positive. Hence, using exactly the same arguments as the proof of \cref{prop: RO-optimal-max}, one can show that the optimal mechanism at this time is ``top-1''.
\end{proof}

\section{Proof of \Cref{thm: structural-barrier-RO}}
Consider the case when there are $n$ players in the game and the new player chooses a fixed strategy $q$. Fix an instance $\omega \in \Omega$. Now, only consider the $n-1$ ``default'' players and the one new player. The expected revenue of the new player is then:
\[
  \sum_{i=0}^{n-1} \binom{n-1}{i}\alpha_{i+1} F(q)^{n-i-1} (1-F(q))^{i} = c(q).
\]
Consider the join of another ``default'' player, the payoff of the new player is always lower in any $\omega$ because his rank could only be unchanged or dropped by one place. Therefore, his expected payoff will only drop. The net payoff is dropped. Hence, it is less than $c(q)$. The only cases that there is no instance that the revenue is decreased is when $q\in \{0,1\}$.

\section{Proof of \Cref{proposition: prop-no-mix-eq}}
\label{app-sec: Proof-of-prop-no-mix-eq}
\NoMixEqProp*
\begin{proof}
    First, we show that the support of any $F_i$ should be a subset of $[0, q_m]$ for some $q_m < 1$. We note that
    \[
     \uPO_i(q_i, q_{-i}) = \frac{q_i}{\sum_{j\in \itemSet}q_j} - c_i(q_i) \leq 1 - c_i(q_i).
    \]
    Since $c_i(0) = 0$ and $\uPO(0, q_{-i}) = 0$, in the NE, by the best response condition, all the producers should at least have $\uPO(q_i ,q_{-i}^\star) \geq 0$ for all $q_i \in \text{supp}(F_i)$. Let $q_{m,i}$ be the values that satisfies $c_i(q_{m,i}) = 1$ and $q_m = \max_{i\in\itemSet} \{q_{m,i}:~i\in \itemSet\}$. We shall have $q_m <  1$ since $c_i(1) > 1$.

    Next, suppose there exists $\F = (F_i)_{i\in \itemSet}$ which is a mixed strategy mixed equilibrium. Then, by the best response condition, the following system of equations should be satisfied
    \[
    \frac{\partial}{\partial q_i}\expectsub{F_{-i}}{u(q_i, q_{-i})} = 
    \expectsub{F_{-i}}{\frac{\partial}{\partial q_i}u(q_i, q_{-i})} = 
    \expectsub{F_{-i}}{\frac{\sum_{j\in \itemSet \setminus i} q_j}{\left(\sum_{k\in \itemSet} q_k\right)^2} - c_i'(q_i)}
    = 0,~~\forall i \in \itemSet,
    \]
    where the expectation is taken over the densities $F_{-i} = (F_j)_{j \in \itemSet\setminus i}$. The first equality is obtained by the bounded convergence theorem and the fact that $c_i'$ is bounded within $[0,q_m]$. By strict convexity of $c_i$, there only exists one $q_i$ such that the above equation is satisfied. Hence, the mixed strategy equilibrium is reduced to the pure strategy equilibrium, as desired.
\end{proof}

\section{Proof of \Cref{proposition: NE-exist-PM}}
\label{app-sec: Proof-of-NE-exist-PM}
We introduce another technical tool to better investigate the equilibrium behaviour. We perform the following variable transformation: Let $\beta^\star = \sum_{i \in \itemSet} q_i^\star$, and $x_i^\star = \frac{q_i^\star}{\sum_{j\in \itemSet} q_j^\star}$. Then, first order conditions for contributing NE could be rewritten as
    \begin{align}
        &x_i^\star + \beta^\star\cdot c_i'(\beta^\star x_i^\star) = 1, ~~\forall i\in \itemSet,  \label{eq: prop-eq-char-no-scale-line1}\\
        &\sum_{i\in \itemSet} \beta^\star c_i'(\beta^\star x_i^\star) = n-1. \label{eq: prop-eq-char-no-scale-line2}
    \end{align}
One obvious benefit of such a variable transformation is that the variables no longer exist in the denominators. Further, we notice that \cref{eq: prop-eq-char-no-scale-line1} corresponds to the first-order optimality conditions of the minimiser of a function $f_{\beta^\star}$, where
\[
f_{\beta}(\x) \defeq  \sum_{i\in \itemSet}  c_i(\beta x_i) + \frac{1}{2}(x_i-1)^2.
\]
And formally, the Nash equilibrium $(\x^\star, \beta^\star)$ satisfies
\[
\x^\star \in \argmin_{\x \in \rr^n} f_{\beta^\star}(\x).
\]
\PMExist*
\begin{proof}

    If $\frac{c_i'(0)}{\sum_{i \in \itemSet} c_i'(0)} \leq \frac{1}{n-1}$. Construct function $G: \mathbb{R}_+ \rightarrow \mathbb{R}^n$ which maps $\beta$ to the root of the system of equations \cref{eq: prop-eq-char-no-scale-line1} parameterised by $\beta$. Indeed, $G$ is continuous. To see this, we fix some $\beta_0 > 0$, and let $x_{0,i}$ satisfy $x_{i,0} + \beta_0 c_i'(\beta_0 x_{i,0}) =1 $. For some $\beta_1$ that $\abs{\beta_1 - \beta_0} \leq \delta$, let $x_i'$ satisfy $x_i' + \beta_1c_i'(\beta_1x_i') = 1$. Since $c_i'$ is differentiable, then it is locally Lipschitz continuous (at the point $\beta_0 x_{i,0}$ with Lipschitz constant $L^+$ (say). Then, we have 
    \begin{align*}
        \abs{x_{i,0} + \beta_1 c_i'(\beta_1 x_{i,0}) - 1} &= \abs{x_{i,0} + \beta_1 c_i'(\beta_1 x_{i,0}) - \left[x_{i,0} + \beta_0c_i'(\beta_0 x_{i,0})\right]} \\
        &= \abs{\beta_1 c_i'(\beta_1 x_{i,0}) - \beta_0c_i'(\beta_0 x_{i,0})}\\
        &\leq \abs{\beta_1 c_i'(\beta_1 x_{i,0}) - \beta c_i'(\beta_0 x_{i,0})} + \abs{\beta_1 c_i'(\beta_0 x_{i,0}) - \beta_0c_i'(\beta_0 x_{i,0})} \\
        &\leq \beta_1\abs{c_i'(\beta_1 x_{i,0}) -  c_i'(\beta_0 x_{i,0})} + \abs{\beta_1 - \beta_0}\cdot c_i'(\beta_0 x_{i,0}) \\
        & \leq \beta_1 L^+ x_{i,0}\abs{\beta_1 - \beta_0} + \abs{\beta_1 - \beta_0}\cdot c_i'(\beta_0 x_{i,0}) \\
        & \leq (\beta_0 + \delta) L^+ x_{i,0}\abs{\beta_1 - \beta_0} + \abs{\beta_1 - \beta_0}\cdot c_i'(\beta_0 x_{i,0}).
    \end{align*}
    And also, we have 
    \begin{align*}
        \abs{x_{i,0} + \beta_1 c_i'(\beta x_{i,0}) - 1} &= \abs{x_{i,0} + \beta_1 c_i'(\beta_1 x_{i,0}) - \left[x_i' + \beta_1 c_i'(\beta_1 x_i') \right]}  \\
        & = \abs{x_{i,0} - x_i' + \beta_1 c_i'(\beta_1 x_{i,0}) - \beta_1 c_i'(\beta_1 x_i')} \\
        & = \abs{x_{i,0} - x_i'} + \abs{\beta_1 c_i'(\beta_1 x_{i,0}) - \beta_1 c_i'(\beta_1 x_i')} \\
        & \geq \abs{x_{i,0} - x_i'},
    \end{align*}
    where the third line follows from the fact that $\abs{x_{i,0} - x_i'}$ and $\abs{\beta_1 c_i'(\beta_1 x_{i,0}) - \beta_1 c_i'(\beta_1 x_i')}$ are of the same signs. Thus, in summary 
    \[
    \abs{x_{i,0} - x_i'} \leq  \abs{x_{i,0} + \beta_1 c_i'(\beta x_{i,0}) - 1} \leq (\beta_0 + \delta) [L^+ x_{i,0} + c_i'(\beta_0 x_{i,0})] \cdot \abs{\beta_1 - \beta_0},
    \]
    which proves that $G$ is continuous. And hence, $\normOne{G}$ is continuous.
    
    It is clear that $G(0) = \mathbf{1}_n$, hence $\normOne{G(0)} = n$. If we set $\beta = \frac{n-1}{\sum_{i\in \itemSet} c_i'(0)}$, then, consider the root of \cref{eq: prop-eq-char-no-scale-line1} $\widehat{\x} \defeq G\left(\frac{n-1}{\sum_{i\in \itemSet} c_i'(0)}\right)$, we could see that
    \[
    \sum_{i\in \mathcal{I}} \widehat{x_i} = n - \sum_{i\in \itemSet}\beta c_i'(\beta \widehat{x_i})  < n - \beta \cdot \sum_{i\in \itemSet} c_i'(0) = n - (n-1) = 1,
    \]
    where the inequality follows from the fact that  $c_i'(\beta \widehat{x_i}) > c_i'(0)$. And also, we notice that, for $\x =0$
    \[
    \widehat{x_i} + \beta c_i'(\beta \widehat{x_i}) = \beta c_i'(0) = \frac{c_i'(0)(n-1)}{\sum_{j\in \itemSet} c_j'(0)} \leq 1.
    \]
    Therefore, it turns out that $\widehat{\x} \geq 0$. In summary, $\normOne{G(0)} = n$ and $\normOne{ G\left(\frac{n-1}{\sum_{i\in \itemSet} c_i'(0)}\right)} < 1$. By the continuity of $\normOne{G(\beta)}$ and the intermediate value theorem, there exists some $\beta^\star \in \left(0,\frac{n-1}{\sum_{i\in \itemSet} c_i'(0)}\right) $ such that $\normOne{G(\beta^\star)} = 1$ and $G(\beta^\star) \geq 0$.
\end{proof}

\section{Proof of \Cref{prop: contributing-subset}}
\label{app-sec: Proof-of-contributing-subset}
\ContributingSubset*
\begin{proof}
    In the equilibrium, we partition the set $\itemSet$ to the set of participating producers $\mathcal{P}$ and the non-participating producers $\mathcal{N}$. For the non-participating producers $i \in \mathcal{N}$ (say), it holds that $0 \in \argmax \uPO_i(q_i^\star, q_{-i}^\star)$. Then, according to \cref{eq: prop-eq-ori}, we should have \begin{equation}
    0 \geq (\uPO_i)'(0, q_{-i}^\star)  = \frac{\sum_{j\in \itemSet \setminus i}q_j^\star}{(\sum_{k\in \itemSet }q_k^\star)^2} - c_i'(0)
     = \frac{\sum_{j\in \mathcal{P}}q_j^\star}{(\sum_{k\in  \mathcal{P} }q_k^\star)^2} - c_i'(0) = \frac{1}{\sum_{k\in  \mathcal{P} }q_k^\star} - c_i'(0). \label{app-eq: non-contributing-condition}
        \end{equation}
    Similarly, for the participating producers $l\in \mathcal{P}$,
    \[
    0 \leq (\uPO_l)'(0, q_{-l}^\star)  = \frac{\sum_{j\in \itemSet \setminus l}q_j^\star}{(\sum_{k\in \itemSet }q_k^\star)^2} - c_l'(0)
     = \frac{\sum_{j\in \mathcal{P}}q_j^\star}{(\sum_{k\in  \mathcal{P} }q_k^\star)^2} - c_l'(0) = \frac{1}{\sum_{k\in  \mathcal{P} }q_k^\star} - c_l'(0).
    \]
    Hence, we have
    \[
    \max_{l\in \mathcal{P}} c_l'(0) \leq \frac{1}{\sum_{k\in  \mathcal{P} }q_k^\star} \leq \min_{i\in \mathcal{N}} c_i'(0).
    \]
    So, we must have $\mathcal{P} = \{1, \ldots, k\}$ for some $k \leq n$. Next, we need to show that $k$ is maximal. Suppose $\mathcal{P} = \{1,\ldots,k\}$ for some non-maximal $k$. Then, the set $\{1, \ldots, k+1\}$ can also facilitate an equilibrium. In this case, consider the equation system \cref{eq: prop-eq-char-no-scale-line1}, \cref{eq: prop-eq-char-no-scale-line2}, both $\mathcal{P}$ and $\mathcal{P} \cup \{k + 1\}$ could lead to a valid solution. However, for the equation system induced by $\mathcal{P} \cup \{k + 1\}$, it holds that $\beta^\star(k+1) \geq \beta^\star(k)$, where $\beta^\star(k)$ indicates the solution $\beta^\star$ of the equation system induced by $\mathcal{P}$. To validate this fact, one can try to plug $\beta^\star(k)$ into the equation system induced by $\mathcal{P} \cup \{k + 1\}$. This leads to \cref{eq: prop-eq-char-no-scale-line2} does not hold. Therefore, $\beta^\star(k+1)$ must be larger (if it is smaller, it also does not make sense since either $x_i$ of the LHS of \cref{eq: prop-eq-char-no-scale-line2} will be even larger).

    Hence, in summary, we since $\beta^\star(k+1) \geq \beta^\star(k)$, by considering \cref{eq: prop-eq-char-no-scale-line1} with index $k+1$, we have $\beta^\star(k+1) c'_{k+1} \leq 1$. Then,
    \[
    c_{k+1}'(0) \leq \frac{1}{\beta^\star(k+1)} \leq  \frac{1}{\beta^\star(k)} = \frac{1}{\sum_{k \in \mathcal{P}} q_k^\star}.
    \]
    where the last equality follows from the fact that $\beta^\star(k) = \sum_{k \in \mathcal{P}} q_k^\star$, which is the change of variable we have used. This contradicts with \cref{app-eq: non-contributing-condition}. 
\end{proof}

\section{Proof of \Cref{cor: entry-barrier-homogeneous}}
\label{app-sec: proof-of-entry-barrier-homogeneous}
\BarrierHomo*
\begin{proof}
    From \cref{eq: prop-eq-char-no-scale-line2}, we have
    \[
    \beta^\star = \frac{n-1}{\sum_{i\in \itemSet} c_i'(\beta^\star x_i^\star)}.
    \]
    Looking at \cref{eq: prop-eq-char-no-scale-line1}, since $x_i > 0$, we observe that 
    \(
      \beta^\star \cdot \bar{c}'(0) < 1.
    \)
    Therefore, $\beta^\star < \frac{1}{\bar{c}'(0)}$ no matter what $n$ is. And the value $\beta^\star x_i^\star < \frac{1}{n \cdot \bar{c}'(0)}$ where the RHS could be arbitrarily small. Therefore, as $n$ goes to infinity, the value 
    \[
     \lim_{n\rightarrow \infty} \frac{1}{\beta^\star}= \lim_{n\rightarrow \infty} \frac{\sum_{i\in \itemSet} c_i'(\beta^\star x_i^\star)}{n-1} = \lim_{n\rightarrow \infty} \frac{n \cdot\bar{c}'(0)}{n-1}  = \bar{c}'(0),
    \]
    as desired.
\end{proof}

\section{Proof of \Cref{lem: EFRM-no-harm}}
\begin{proposition}\label{prop:symmetric-char-low-cost}
Suppose $\alpha_1 > \alpha_n \ge c(0)$. Let $F$ be the continuous c.d.f.~below:
\[
F(q) ~=~ \begin{cases}
0, & \text{if }q\le 0;\\
\Rinv(c(q)-c(0)+\alpha_n), & \text{if }0 < q < c^{-1}(\alpha_1 - \alpha_n + c(0));\\
1, & \text{if }q\ge c^{-1}(\alpha_1 - \alpha_n + c(0)).
\end{cases}
\]
Then $(F,F,\ldots,F)$ is a symmetric NE.
\end{proposition}

\begin{proof}    First, note that $R(F(q))$ is the expected reward received by a player at $(F,F,\ldots,F)$ when she chooses quality measure $q$.
Let $q_{\max} = c^{-1}(\alpha_1 - \alpha_n + c(0))$. Observe that $R(F(0)) = R(0) = \alpha_n$, and $R(F(q_{\max})) = R(1) = \alpha_1$.
The expected utility $u(q) = R(q) - c(q) = \alpha_n - c(0) \ge 0$ when $0\le q \le q_{\max}$.
When $q > q_{\max}$, the expected utility is $\alpha_1 - c(q) < \alpha_1 - c(q_{\max})= \alpha_n - c(0)$.
Thus, $u(q)$ is maximized in the interval $0\le q \le q_{\max}$, which matches with the support of $F$.
\end{proof}
Since the EFRM always ensures $\alpha_n = c(0) = \xi$. The equilibrium structure will not be changed if $\xi$ is charged and given back trivially.

\section{Proof of \Cref{thm: optimal-subsidy-scheme}}
Both results are immediate from \Cref{app-lem: RO-action-improve-mean} and \Cref{app-lem: RO-action-improve-max}. Specifically, by \Cref{prop:symmetric-char-low-cost}, there is no change to the equilibrium structure rather than the flexibility to reassign the entry fees. \Cref{app-lem: RO-action-improve-mean} pointed out that the max-min mechanism is the optimal way and \Cref{app-lem: RO-action-improve-max} pointed out that the max-max mechanism is the optimal way.

\section{Proof of \Cref{prop: goodness-of-commision}}
\label{app-sec: proof-goodness-of-commission}
\GoodnessOfCommission*
\begin{proof}
    First, from \cref{eq: prop-eq-char-no-scale-line1}, since $c_i' > c_j'$, we must have $x_i^\star < x_j^\star$ at the equilibrium (which means $q_i^\star < q_j^\star$).
     Consider the profit $u_i$,
    \begin{align*}
     u_i(q_i^\star, \mathbf{q}^\star_{-i}) 
     = \frac{q_i^\star}{\sum_{k\in \itemSet} q_k^\star} - c_i(q_i^\star)
     &<  \frac{q_i^\star}{q_i^\star + q_j^\star  + \sum_{k\in \itemSet \setminus \{i,j\}} q_k^\star} - c_j(q_i^\star) \\
     &\leq \frac{q_i^\star}{q_i^\star + q_i^\star  + \sum_{k\in \itemSet \setminus \{i,j\}} q_k^\star} - c_j(q_i^\star) \\
     &\leq \frac{q_j^\star}{q_j^\star + q_i^\star  + \sum_{k\in \itemSet \setminus \{i,j\}} q_k^\star} - c_j(q_j^\star)
     =  u_j(q_j^\star, \mathbf{q}^\star_{-j}).
    \end{align*}
    The first line follows from $c_i > c_j$, the second line follows from $q_i^\star < q_j^\star$ and the last line follows from the optimality of $q_j^\star$.
\end{proof}

\end{document}